\documentclass[conference]{IEEEtran}
\IEEEoverridecommandlockouts

\usepackage[T1]{fontenc}
\usepackage{aecompl}
\usepackage{graphics, graphicx}
\usepackage{cite, svg, epsfig, enumerate}
\usepackage{amssymb, amsmath,bm, multirow, amsmath}
\usepackage{array, nicematrix}

\newtheorem{theorem}{Theorem}

\newtheorem{lemma}{Lemma}
\newtheorem{proposition}{Proposition}
\newtheorem{definition}{Definition}

\newcommand{\overbar}[1]{\mkern 1.5mu\overline{\mkern-1.5mu#1\mkern-1.5mu}\mkern 1.5mu}

\DeclareMathOperator{\inte}{\textrm{Int}}

\begin{document}

\title{Optimal Safety Control using High-Order Control Barrier Functions\\
\thanks{This work was supported by the National Natural Science Foundation of China under Grants 62573085, the Fundamental Research Funds for the Central Universities under Grant DUT22RT(3)090, and the LiaoNing Revitalization Talents Program (XLYC2403048).}
}

\author{Neng Li, Zuodong Pan, Jiaxing Wang, Weiguo Xia, and Wei Ren
\thanks{N. Li, Z. Pan J. Wang, W. Xia and W. Ren are with the Key Laboratory of Intelligent Control and Optimization for Industrial Equipment of Ministry of Education, Dalian University of Technology, Dalian 116024, China. 
Email: \textrm{\small \{744845870, zuodongpan, wanna\}@mail.dlut.edu.cn, \{wgxiaseu, wei.ren\}@dlut.edu.cn)}.}
}

\maketitle

\begin{abstract}
This paper investigates the optimal safety control problem of nonlinear control systems by proposing novel high-order control barrier functions (HOCBFs). Different from zeroing HOCBFs, two novel HOCBFs are derived and the safety controllers are designed in an explicit way. Next, we implement vector Lyapunov function approach to propose a novel high-order control Lyapunov function (HOCLF) for the stabilization control problem. The relations between the proposed and existing HOCBFs are discussed. Afterwards, the compatibility of the proposed HOCLF and HOCBF is addressed to guarantee the stabilization and safety control objectives simultaneously, and thus the optimal controller is established. Finally, a numerical example from the navigation problem of quadrotors is presented to illustrate the efficacy of the derived results. 
\end{abstract}


\section{Introduction}
\label{sec-introduction}

To describe and guarantee different system performances, e.g., stability, safety and robustness, numerous certificate functions have been proposed and extensively applied in the past few decades \cite{Dawson2023safe}. Among these certificate functions, the most well-known may be Lyapunov functions, whose strict decrease along the state trajectory is able to show the convergence of the system state to the equilibrium \cite{Khalil2002nonlinear}. Hence, Lyapunov functions have been applied extensively \cite{Dixon2003nonlinear} and further extended as control Lyapunov functions (CLFs) for the control design \cite{Zvi1983stabilization}. Lyapunov functions enable us to avoid to compute the exact solutions of dynamical systems \cite{Dixon2003nonlinear}, and thus facilitate both stability analysis and controller design. In a similar way, other certificate functions are constructed, e.g., control barrier functions (CBFs) for system safety \cite{Ames2016control, Ren2022razumikhin}, artificial potential functions for safe convergence \cite{Paternain2019stochastic}, navigation functions for motion coordination \cite{Tanner2005formation}, and contract metrics for trajectory tracking \cite{Dawson2023safe}. These certificate functions are useful both theoretically and practically such that nonlinear and complex systems can be addressed and the computational complexity can be reduced. For instance, with CBFs, the computation of the reachable sets of dynamical systems is avoided and safety controllers are derived explicitly \cite{Prajna2004safety, Wieland2007constructive}. 

With increasing integrated demands, different system performances are required to be satisfied simultaneously, and thus different certificate functions are combined either implicitly or explicitly. A fundamental integration is for the stabilization and safety, both of which are the basic performances of dynamical systems. Hence, both CLFs and CBFs have been combined effectively in numerous applications, including automotive control problems \cite{Panagou2015distributed, Ames2016control}. In particular, both zeroing and reciprocal CBFs are proposed \cite{Ames2016control} such that the combinations can be done from different perspectives. However, many existing results are devoted to the classic first-order cases, while are not available for the so-called high-order cases. 

In terms of practical applications, dynamical systems are usually governed by physical laws, including Newton's law, Theorem of linear and angular momentum, and Kirchhoff's laws of current and voltage \cite{Siciliano2008springer}. Hence, the real-world system models are of second and higher orders. For instance, the dynamics of quadrotors are based on the Euler-Lagrange equation \cite{lee2010geometric} and are of second order. For the high-order cases, the existing CLFs/CBFs are not available, since the Lie derivatives of CLFs/CBFs related to control inputs are required to be non-zero. These requirements are to guarantee the existence of the controllers and are violated in the high-order cases \cite{Jankovic2018robust}. For instance, the control force of quadrotors is imposed on the angular velocity, while the desired performances are on the position space. In this way, if the CLF is defined on the position space, then the Lie derivative is zero such that the controller cannot be designed. As a result, numerous attention has been paid to the high-order cases \cite{Tan2021high, Xiao2021high, Xu2018constrained, Ong2024rectified}, and aims to propose a systematic framework which is highly relevant for real-world applications. However, the existing results are based on zeroing CBFs, which are from the function describing the safe set. Note that not all the description functions can be taken directly as CBFs, and for the high-order cases the constructed safe set is much smaller than the original one, which inevitably results in the conservatism. 

Inspired by the above discussions, in this paper we propose novel HOCBFs to address the optimal safety control problem of nonlinear control systems. Different from zeroing HOCBFs \cite{Tan2021high, Xiao2021high}, we first propose a reciprocal HOCBFs, which offers more flexibility in terms of the construction and applications, and design an explicit safety controller. Next, based on vector CLFs \cite{Ren2023vector}, a novel high-order CLF (HOCLF) is developed to address the stabilization control problem. Following the similar mechanism, a novel HOCBF is proposed such that the safe set does not need to be reconstructed, thereby relaxing the conservatism from zeroing HOCBFs. Note that with the proposed HOCBF and HOCLF, both the stabilizing and safety controllers can be derived explicitly, and the relations between the proposed and existing functions are discussed. Finally, the compatibility of the proposed HOCBF and HOCLF is investigated to satisfy the stabilization and safety objectives simultaneously. In particular, an HOCBF-stabilizable set is proposed to mediate the conflict between the two objectives, which extends the result in \cite{Cortez2022compatibility} to the high-order cases, and further the optimal safety controller is established, which is illustrated via a numerical example. 

The remainder of this paper is organized as follows. Preliminaries are presented in Section \ref{sec-consys}. High-order control functions are proposed in Section \ref{sec-HOCF}. The optimal controller is designed in Section \ref{sec-optcon}. Numerical results are presented in Section \ref{sec-example}, followed by the conclusion and future works in Section \ref{sec-conclusion}.
 
\section{Problem Formulation}
\label{sec-consys}

Let $\mathbb{R}:=(-\infty, +\infty), \mathbb{R}_{+}:=[0, +\infty)$ and $\mathbb{N}:=\{1, 2, \ldots\}$. For $x\in\mathbb{R}^{n}$, $x_{i}$ is the $i$-th element of $x$; $|x|$ is the Euclidean norm of $x$. For $x, y\in\mathbb{R}^{n}$, $(x, y):=(x^{\top}, y^{\top})^{\top}$, and $x\succeq y$ (or $x\preceq y$) if $x_{i}\geq y_{i}$ (or $x_{i}\leq y_{i}$) for all $i=1, \ldots, n$. A matrix $A\in\mathbb{R}^{n\times n}$ is \emph{Hurwitz} if all its eigenvalues lie in the open left half of the complex plane. Given a set $\mathbb{S}\subset\mathbb{R}^{n}$, $\partial\mathbb{S}$ is the boundary of $\mathbb{S}$; $\inte(\mathbb{S})$ is the interior of $\mathbb{S}$; $\hat{\mathbb{S}}$ is the closure of $\mathbb{S}$. Given $\delta>0$ and $\mathbf{x}\in\mathbb{R}^{n}$, the open ball centered at $\mathbf{x}$ with radius $\delta$ is defined as $\mathbb{B}(\mathbf{x}, \delta):=\{x\in\mathbb{R}^{n}: |x-\mathbf{x}|<\delta\}$; $\mathbb{B}(\delta):=\mathbb{B}(0, \delta)$. A continuous function $\alpha: \mathbb{R}_{+}\rightarrow\mathbb{R}_{+}$ is of class $\mathcal{K}$, if it is strictly increasing and  $\alpha(0)=0$; it is of class $\mathcal{K}_{\infty}$, if it is of class $\mathcal{K}$ and unbounded. A continuous function $\alpha: \mathbb{R}\rightarrow\mathbb{R}$ is of extended class $\mathcal{K}_{e}$, if it is strictly increasing and $\alpha(0)=0$. A continuous function $\beta: \mathbb{R}_{+}\times\mathbb{R}_{+}\rightarrow\mathbb{R}_{+}$ is of class $\mathcal{KL}$, if $\beta(s, t)\in\mathcal{K}$ for each fixed $t\geq0$ and $\beta(s, t)$ decreases to zero as $t\rightarrow\infty$ for each fixed $s\geq0$. $\mathfrak{C}(\mathbb{R}^{n}, \mathbb{R}^{p})$ denotes the class of continuous functions mapping $\mathbb{R}^{n}$ to $\mathbb{R}^{p}$. For $r\in\mathbb{N}$, $\mathfrak{C}^{r}(\mathbb{R}^{n}, \mathbb{R}^{p})$ denotes the class of continuously $r$-th differentiable functions mapping $\mathbb{R}^{n}$ to $\mathbb{R}^{p}$.

Consider the following nonlinear control system
\begin{equation}
\label{eqn-1}
\dot{x}=f(x)+g(x)u, 
\end{equation}
where $x\in\mathbb{R}^{n}$ is the state initialized at $x_{0}\in\mathbb{X}_{0}\subset\mathbb{R}^{n}$, and $u\in\mathbb{U}\subset\mathbb{R}^{m}$ is the control input. Both $f: \mathbb{R}^{n}\rightarrow\mathbb{R}^{n}$ and $g: \mathbb{R}^{n}\rightarrow\mathbb{R}^{n\times m}$ are assumed to be locally Lipschitz  continuous, which ensures the existence of the unique solution to \eqref{eqn-1}; see \cite{Khalil2002nonlinear} for more details. Let $f(0)=0$ and $g(0)=0$. Hence, $x(t)\equiv0$ with $t\geq0$ is a trivial solution to \eqref{eqn-1}.

\begin{definition}[\cite{Khalil2002nonlinear}]
\label{def-1}
Given a control input $u\in\mathbb{R}^{m}$, the system \eqref{eqn-1} is \emph{globally asymptotically stable (GAS)}, if there exists $\beta\in\mathcal{KL}$ such that for all $x_{0}\in\mathbb{R}^{n}$,
\begin{align}
\label{eqn-2}
|x(t)|&\leq\beta(|x_{0}|, t), \quad \forall t\geq0.
\end{align}
In addition, if \eqref{eqn-2} holds for $x_{0}\in\mathbb{X}_{0}\subset\mathbb{R}^{n}$, then the system \eqref{eqn-1} is semiglobally asymptotically stable (SGAS). 
\end{definition}

\begin{definition}[\cite{Ames2016control}]
\label{def-2}
A set $\mathbb{S}\subset\mathbb{R}^{n}$ is \emph{forward invariant (FI)} for the system \eqref{eqn-1}, if $x(t)\in\mathbb{S}$ for any trajectory $x(t)$ starting from $x_{0}\in\mathbb{S}$. If $\mathbb{S}$ is FI, then the system \eqref{eqn-1} is \emph{safe} with respect to $\mathbb{S}$, and $\mathbb{S}$ is called the \emph{safe set}. 
\end{definition}

A set $\mathbb{S}\subset\mathbb{R}^{n}$ is associated with a function $h\in\mathfrak{C}(\mathbb{R}^{n}, \mathbb{R})$, and thus the set $\mathbb{S}\subset\mathbb{R}^{n}$ can be described explicitly as 
\begin{align}
\label{eqn-3}
\mathbb{S}:=\left\{x\in\mathbb{R}^{n}: h(x)\geq0\right\}.
\end{align}
Let $\partial\mathbb{S}:=\{x\in\mathbb{R}^{n}: h(x)=0\}$ and $\inte(\mathbb{S}):=\{x\in\mathbb{R}^{n}: h(x)>0\}$. Assume that $\inte(\mathbb{S})\neq\varnothing$ and $\widehat{\inte(\mathbb{S})}=\mathbb{S}$.

For the system \eqref{eqn-1}, we aims to design optimal controllers to ensure the systems stabilization and safety. To this end, both CLF $V\in\mathfrak{C}(\mathbb{R}^{n}, \mathbb{R}_{+})$ and CBF $B\in\mathfrak{C}(\mathbb{R}^{n}, \mathbb{R})$ are proposed and combined together via optimization problems; see, e.g., \cite{Ames2016control, Jankovic2018robust}. In this way, the optimal controllers can be designed. However, the Lie derivatives $L_{g}V(x):=\frac{\partial V(x)}{\partial x}g(x)$ and $L_{g}B(x):=\frac{\partial B(x)}{\partial x}g(x)$ are required to be non-zero such that the optimization problems can be solved, thereby resulting in additional constraints on the applicability of the CLF/CBF and the availability of dynamical systems. Hence, in this paper novel high-order CBFs and CLF are proposed to guarantee the stabilization and safety objectives. 

\section{High-Order Control Lyapunov and Barrier Functions}
\label{sec-HOCF}

In this section, novel high-order control Lyapunov and barrier functions are proposed and further implemented to design the safety and stabilizing controllers in an explicit way. 

\subsection{High-Order Control Lyapunov Function}
\label{subsec-HOCLF}

Motivated by vector CLFs in \cite{Ren2023vector}, a novel HOCLF is proposed for the stabilization control in the high-order cases. 

\begin{definition}
\label{def-5}
Consider the system \eqref{eqn-1}. A function $V\in\mathfrak{C}^{p}(\mathbb{R}^{n},$ $\mathbb{R}_{+})$ with $p\geq2$ is called a \emph{high-order control Lyapunov function (HOCLF)}, if for all $x\in\mathbb{R}^{n}$, 
\begin{enumerate}[(i)]
  \item $L_{g}V(x)=\ldots=L_{g}L^{p-2}_{f}V(x)=0$, $L_{g}L^{p-1}_{f}V(x)\neq0$, and there exist $\alpha_{1}, \alpha_{2}\in\mathcal{K}_{\infty}$ such that
  \begin{equation}
  \label{eqn-12}
  \alpha_{1}(|x|)\leq V(x)\leq\alpha_{2}(|x|); 
  \end{equation}
  
  \item  there exists $\Lambda=(\lambda_{1}, \ldots, \lambda_{p})\in\mathbb{R}^{p\times p}$ such that
  \begin{equation}
  \label{eqn-13}
  \inf\nolimits_{u\in\mathbb{U}}\left\{L_{f}\mathbf{V}(x)+L_{g}\mathbf{V}(x)u\right\}\preceq\Lambda\mathbf{V}(x), 
  \end{equation}
  where $\mathbf{V}(x):=(V(x), L_{f}V(x), \ldots, L^{p-1}_{f}V(x))$, $L_{f}\mathbf{V}(x):=(L_{f}V(x), L^{2}_{f}V(x), \ldots, L^{p}_{f}V(x))\in\mathbb{R}^{p}$ and $L_{g}\mathbf{V}(x):=(0, \ldots,$ $0, L_{g}L^{p-1}_{f}V(x))\in\mathbb{R}^{p\times m}$;
  
  \item the matrix $\Lambda\in\mathbb{R}^{p\times p}$ is Hurwitz, and there exist orthogonal matrices $\Omega:=(\omega_{1}, \ldots, \omega_{p})\in\mathbb{R}^{p\times p}, \Pi:=(\pi_{1}, \ldots, \pi_{r})\in\mathbb{R}^{p\times p}$ and a Jordan matrix $\bar{\Lambda}\in\mathbb{R}^{p\times p}$ such that $\Lambda=\Omega^{\top}\bar{\Lambda}\Pi$ and $\sum^{p}_{i=1}\omega^{\top}_{1}\pi_{i}\mathbf{V}_{i}(x_{0})>0$.
\end{enumerate}
\end{definition}

Definition \ref{def-5} is different from the one in \cite{Frauenfelder2023decentralized} based on the same mechanism as in \cite{Xiao2021high} for HOCBFs. In Definition \ref{def-5}, all the derivatives of $V$ are coupled via a vector inequality \eqref{eqn-13}, while the lasso-like functions are involved in \cite{Frauenfelder2023decentralized}. The matrix $\Lambda\in\mathbb{R}^{p\times p}$ in \eqref{eqn-13} needs to satisfy item (iii), whose necessity is presented in the proof of Theorem \ref{thm-2}. A special case is 
\begin{align}
\label{eqn-14}
\Lambda=\begin{bmatrix}
\lambda^{\top}_{1}\\\vdots \\ \lambda^{\top}_{r}
\end{bmatrix}=\begin{bNiceArray}{c|cw{c}{1cm}c}[margin]
0& \Block{3-3}{I} & &  \\
\vdots & & &  \\
0 & & & \\
\hline
\lambda_{p1} &\lambda_{p2} & \cdots& \lambda_{pp}
\end{bNiceArray}, 
\end{align}
which is based on the nature of $\mathbf{V}(x)$ and only imposes the constraint on the derivative of $L^{p-1}_{f}V(x)$. Based on the proposed HOCLF, the following result is derived to design the stabilizing controller for the high-order cases. 

\begin{theorem}
\label{thm-2}
If the system \eqref{eqn-1} admits an HOCLF $V\in\mathfrak{C}^{p}(\mathbb{R}^{n},$ $\mathbb{R}_{+})$, then the closed-loop system is GAS under the controller
\begin{align}
\label{eqn-15}
u(x):=\left\{\begin{aligned}
&0, &&\text{if\ } x=0, \\
&\kappa(\mathfrak{a}(x), \mathfrak{b}(x)), &&  \text{otherwise}, 
\end{aligned}\right.
\end{align}
where $\mathfrak{a}(x):=L^{p}_{f}V(x)-\lambda^{\top}_{p}\mathbf{V}(x), \mathfrak{b}(x):=L_{g}L^{p-1}_{f}V(x)$, and 
\begin{equation}
\label{eqn-11}
\kappa(\mathfrak{a}(x), \mathfrak{b}(x)):=\frac{\mathfrak{a}(x)+\sqrt{\mathfrak{a}^{2}(x)+|\mathfrak{b}(x)|^{4}}}{-|\mathfrak{b}(x)|^{2}}\mathfrak{b}^{\top}(x).
\end{equation}
\end{theorem}

\begin{IEEEproof}
From the controller \eqref{eqn-15}, if $x=0$, then $u(x)=0$ and thus $L^{p}_{f}V(x)+L_{g}L^{p-1}_{f}V(x)u-\lambda^{\top}_{p}\mathbf{V}(x)=0$; otherwise, 
\begin{align*}
&L^{p}_{f}V(x)+L_{g}L^{p-1}_{f}V(x)u-\lambda^{\top}_{p}\mathbf{V}(x) \\
&=\mathfrak{a}(x)-\mathfrak{d}(x)\frac{\mathfrak{a}(x)\mathfrak{b}^{\top}(x)}{|\mathfrak{b}(x)|^{2}}-\mathfrak{b}(x)\frac{\sqrt{\mathfrak{a}^{2}(x)+|\mathfrak{b}(x)|^{4}}}{|\mathfrak{b}(x)|^{2}}\mathfrak{b}^{\top}(x) \\
&=-\sqrt{\mathfrak{a}^{2}(x)+|\mathfrak{b}(x)|^{4}}\leq0.
\end{align*}
In addition, $L_{g}V(x)=\ldots=L_{g}L^{p-2}_{f}V(x)=0$ and $L_{g}L^{p-1}_{f}V(x)\neq0$. Hence, for all $x\in\mathbb{R}^{n}$, $L_{f}\mathbf{V}(x)+L_{g}\mathbf{V}(x)u\preceq\Lambda\mathbf{V}(x)$, that is, $\dot{\mathbf{V}}(x)\preceq\Lambda\mathbf{V}(x)$. 

We define the following comparison system:
\begin{align*}
\dot{\xi}(t)&=\Lambda\xi(t), \quad  \xi(0)=\mathbf{V}(x_{0})\in\mathbb{R}^{p}.
\end{align*}
From item (iii) of Definition \ref{def-5}, the first component of $\xi(t)$ satisfies $\xi_{1}(t)\leq e^{\lambda_{\max}(\bar{\Lambda})t}\sum^{p}_{i=1}\omega^{\top}_{1}\pi_{i}\xi_{i}(0)$ and $\lambda_{\max}(\bar{\Lambda})<0$, where $\lambda_{\max}(\bar{\Lambda})$ is the maximal eigenvalue of $\bar{\Lambda}$. That is, $\xi(t)$ decreases to zero as $t\rightarrow\infty$ asymptotically. From the comparison principle in \cite[Lem. 1]{Ren2018vector}, $V(x)\leq\xi_{1}(t)$ for all $t>0$, and thus the closed-loop system is GAS. 
\end{IEEEproof}

From the proof of Theorem \ref{thm-2} we can the necessity of item (iii) of Definition \ref{def-5}. Note that $V(x_{0})>0$ is true, while whether $L_{f}V(x_{0}), \ldots, L^{p-1}_{f}V(x_{0})$ are positive is unknown. Hence, item (iii) of Definition \ref{def-5} is needed for the convergence of $V(x)$ and constrains the choice of the initial state. On the other, the controller \eqref{eqn-15} is continuous in the region away from the origin, and can be continuous at the origin by introducing the high-order small control property (HOSCP): in a small region around the origin, there exists a bounded control input such that \eqref{eqn-13} is satisfied. With the HOSCP, we can follow \cite{Sontag1989universal} to show the continuity of the controller \eqref{eqn-15} at the origin. 

\subsection{Reciprocal High-Order Control Barrier Function}
\label{subsec-HOCBF}

Regarding the safety control, the existing HOCBFs are recalled first, and then a reciprocal HOCBF is proposed.  

\begin{definition}[\cite{Tan2021high}]
\label{def-3}
Consider the system \eqref{eqn-1} and a set $\mathbb{D}\subset\mathbb{R}^{n}$. A function $h\in\mathfrak{C}^{r}(\mathbb{R}^{n}, \mathbb{R})$ has \emph{least relative degree} $r\in\mathbb{N}$ in $\mathbb{D}$ for the system \eqref{eqn-1}, if $L_{g}L^{k}_{f}h(x)\equiv0$ for all $x\in\mathbb{D}$ and $k=0, \ldots, r-2$, where $L^{k}_{f}h(x):=\frac{\partial^{k} h(x)}{\partial x^{k}}f(x)$.
\end{definition}

If the function $h$ in \eqref{eqn-3} has least relative degree $r\in\mathbb{N}$, then $h$ can be defined as the high-order control barrier function (HOCBF) via the following mechanism \cite{Tan2021high, Xiao2021high}. First, for all $i=1, \ldots, r$,  we define the functions:
\begin{align}
\label{eqn-4}
\psi_{0}(x)&:=h(x), \quad \psi_{i}(x):=\left(\frac{d}{dt}+\gamma_{i}\right)\psi_{i-1}(x), 
\end{align}
where $\gamma_{i}\in\mathcal{K}_{e}$. Next, the following sets are defined:
\begin{align}
\label{eqn-5}
\mathbb{S}_{i}&:=\{x\in\mathbb{R}^{n}: \psi_{i-1}(x)\geq0\}, \quad i=1, \ldots, r,  \\
\label{eqn-6}
\overbar{\mathbb{S}}&:=\cap^{r}_{i=1}\mathbb{S}_{i}. 
\end{align}
In the high-order cases, $\overbar{\mathbb{S}}$ is taken as the \emph{safe set}. From \eqref{eqn-3}-\eqref{eqn-6}, we can see that $\overbar{\mathbb{S}}\subseteq\mathbb{S}_{i}$ and $\overbar{\mathbb{S}}\subseteq\mathbb{S}$. That is, a smaller safe set is addressed for the high-order cases. Finally, from \cite{Xiao2021high}, $h$ is an HOCBF, if there exists $\mathbb{D}\subset\mathbb{R}^{n}$ such that $\overbar{\mathbb{S}}\subset\mathbb{D}$ and $\psi_{r}(x)\geq0$ for all $x\in\mathbb{D}$. If the HOCBF $h$ does exist, then the set $\overbar{\mathbb{S}}$ is guaranteed to be FI. 

From the above mechanism, the HOCBF is based on the function describing the safe set, and thus $h$ is called the \emph{zeroing HOCBF}; see also \cite{Xiao2021high}. Note that $h$ may be simple and cannot be differentiable appropriately; see \cite{Ames2016control} for examples. Another way to introduce the HOCBF is based on the function $h$ implicitly, and the resulting HOCBF is called the \emph{reciprocal HOCBF}. Following this direction, we define the following functions. For all $i=1, \ldots, r$, 
\begin{align}
\label{eqn-7}
\hspace{-10pt}\phi_{0}(x)&:=B(x), \  \frac{1}{\phi_{i}(x)}:=\eta_{i}\left(\frac{1}{\phi_{i-1}(x)}\right)-\frac{d\phi_{i-1}(x)}{dt}, 
\end{align}
where $B\in\mathfrak{C}(\mathbb{R}^{n}, \mathbb{R})$ is a function related to $h$, and $\eta_{i}\in\mathcal{K}_{e}$. From \eqref{eqn-7}, a novel HOCBF is proposed below. 

\begin{definition}
\label{def-4}
Consider the system \eqref{eqn-1} and the set $\overbar{\mathbb{S}}$ in \eqref{eqn-6}. A function $B\in\mathfrak{C}^{r}(\inte(\overbar{\mathbb{S}}), \mathbb{R})$ is called a \emph{high-order control barrier function (HOCBF-I)}, if for all $x\in\inte(\overbar{\mathbb{S}})$,
\begin{enumerate}[(i)]
  \item $L_{g}L^{r-1}_{f}B(x)\neq0$ and $L_{g}L^{k}_{f}B(x)=0$, $k=0, \ldots, r-2$;

  \item  there exist $\alpha_{i1}, \alpha_{i2}\in\mathcal{K}_{e}$ such that for $i=0, \ldots, r-1$, 
  \begin{align}
  \label{eqn-8}
  \hspace{-20pt}\alpha_{i1}(\psi_{i}(x))\leq\frac{1}{\phi_{i}(x)}&\leq\alpha_{i2}(\psi_{i}(x)), \\
  \label{eqn-9}
  \hspace{-20pt}\inf_{u\in\mathbb{U}}\{L_{f}\phi_{r-1}(x)+L_{g}\phi_{r-1}(x)u\}&\leq\eta_{r}\left(\frac{1}{\phi_{r-1}(x)}\right), 
  \end{align}
  where $\eta_{r}$ is given in \eqref{eqn-7}. 
\end{enumerate}
\end{definition}

Definition \ref{def-4} extends the reciprocal CBF in \cite{Ames2016control} to the high-order cases. Item (i) is from $h(x)$ of least relative order $r$. In particular, if $B(x)$ is defined via $h(x)$ (see, e.g., \cite{Ames2016control}), then let $B(x):=\mathfrak{F}(h(x))$ with $\mathfrak{F}\in\mathfrak{C}^{r}(\mathbb{R}, \mathbb{R})$, and thus $L_{g}B(x)=\frac{\partial\mathfrak{F}(h)}{\partial h}L_{g}h(x)=0$. Based on the chain rule for derivatives and the property of $h(x)$, we can derive $L_{g}B(x)=\ldots=L_{g}L^{r-2}_{f}B(x)=0$ and $L_{g}L^{r-1}_{f}B(x)\neq0$. In item (ii), \eqref{eqn-8} shows the relations between $\phi_{i}(x)$ and $\psi_{i}(x)$. If $i=0$, then \eqref{eqn-8} is the same as the one in \cite{Ames2016control} to show the relation between $B(x)$ and $h(x)$. \eqref{eqn-9} is for the $r$-th order case and imposes the constraint on the control input. From the definition of $\phi_{i}(x)$, we can see that \eqref{eqn-9} equals to $\dot{\phi}_{r-1}(x)\leq\eta_{r}(1/\phi_{r-1}(x))$, and thus $\phi_{r}(x)\geq0$. The next theorem presents the safety controller design from the HOCBF-I. 

\begin{theorem}
\label{thm-1}
Consider the system \eqref{eqn-1} admitting an HOCBF-I $B\in\mathfrak{C}^{r}(\inte(\overbar{\mathbb{S}}), \mathbb{R})$ with the set $\overbar{\mathbb{S}}$ in \eqref{eqn-6}. The system \eqref{eqn-1} is safe with respect to $\overbar{\mathbb{S}}$ under the following controller
\begin{align}
\label{eqn-10}
u(x):=\left\{\begin{aligned}
&0, &&\text{if\ } x\equiv0, \\
&\kappa(\mathfrak{c}(x), \mathfrak{d}(x)), &&  \text{otherwise}, 
\end{aligned}\right.
\end{align}
where $\mathfrak{c}(x):=L_{f}\phi_{r-1}(x)-\eta_{r}(1/\phi_{r-1}(x))$, $\mathfrak{d}(x):=L_{g}\phi_{r-1}(x)$ and $\kappa$ is defined in \eqref{eqn-11}.
\end{theorem}

\begin{IEEEproof}
First, let $\mathsf{H}_{r-1}(x):=1/\phi_{r-1}(x)$, and  we have 
\begin{align*}
\dot{\mathsf{H}}_{r-1}(x)&=-(L_{f}\phi_{r-1}(x)+L_{g}\phi_{r-1}(x)u)/\phi^{2}_{r-1}(x).
\end{align*}
From \eqref{eqn-9}, we have
\begin{align*}
\dot{\mathsf{H}}_{r-1}(x)&\geq-\eta_{r}(1/\phi_{r-1}(x))/\phi^{2}_{r-1}(x) \\
&=-\mathsf{H}^{2}_{r-1}(x)\eta_{r}(\mathsf{H}_{r-1}(x)) \\
&=:-\bar{\eta}_{r}(\mathsf{H}_{r-1}(x)), 
\end{align*}
where $\bar{\eta}_{r}$ is a continuous function. In addition, $\bar{\eta}_{r}(0)=0$ and $\bar{\eta}_{r}(\mathsf{H}_{r-1}(x))>0$ if and only if $\mathsf{H}_{r-1}(x)>0$. Since $x_{0}\in\inte(\overbar{\mathbb{S}})$, $\mathsf{H}_{r-1}(x_{0})>0$ holds from \eqref{eqn-8} of Definition \ref{def-4}. From the comparison principle \cite[Sec. 3.4]{Khalil2002nonlinear} and \cite[Lem. 4.4]{Khalil2002nonlinear}, there exists $\beta\in\mathcal{KL}$ such that $\mathsf{H}_{r-1}(x)\geq\beta(\mathsf{H}_{r-1}(x_{0}), t)$ for all $t\geq0$. From \eqref{eqn-8}, we have 
\begin{align*}
\psi_{r-1}(x)\geq\alpha^{-1}_{(r-1)2}(\beta(\alpha_{(r-1)1}(\psi_{r-1}(x_{0})), t)). 
\end{align*}
Hence, $\psi_{r-1}(x)>0$. That is, $\inte(\mathbb{S}_{r})\neq\varnothing$ and thus $\mathbb{S}_{r}$ is FI. 

Next, since $\mathsf{H}_{r-1}(x)>0$ for all $x\in\inte(\overbar{\mathbb{S}})$, from the definitions of $\mathsf{H}_{r-1}(x)$ and $\phi_{r-1}(x)$, we have 
\begin{align*}
\dot{\phi}_{r-2}(x)<\eta_{r-1}(1/\phi_{r-2}(x)). 
\end{align*}
Let $\mathsf{H}_{r-2}(x):=1/\phi_{r-2}(x)$, and thus 
\begin{align*}
\dot{\mathsf{H}}_{r-2}(x)&=-\dot{\phi}_{r-2}(x)/\phi^{2}_{r-2}(x)\\
&>-\mathsf{H}^{2}_{r-2}(x)\eta_{r-1}(\mathsf{H}_{r-2}(x)). 
\end{align*}
Following the same mechanism as in the case of $r-1$, we can show that $\inte(\mathbb{S}_{r-1})\neq\varnothing$ and $\mathbb{S}_{r-1}$ is FI. 

Finally, the above mechanism can be implemented iteratively, and thus for all $i=1, \ldots, r$, $\inte(\mathbb{S}_{i})\neq\varnothing$ and $\mathbb{S}_{i}$ is FI. Hence, the set $\inte(\overbar{\mathbb{S}})$ is FI under the controller \eqref{eqn-10}. 
\end{IEEEproof}

Theorem \ref{thm-1} shows how to design the safety controller from the HOCBF-I in Definition \ref{def-4}, and thus extends the existing results in \cite{Tan2021high, Xiao2021high} from the zeroing type to the reciprocal type. Hence, the HOCBF is allowed to be different from $h$ in \eqref{eqn-3}. On the other hand, similar to \cite{Tan2021high, Xiao2021high}, here the FI property is for the set $\overbar{\mathbb{S}}$ in \eqref{eqn-6}, which is a subset of $\mathbb{S}$ in \eqref{eqn-3}. How to guarantee the FI property of the set $\mathbb{S}$ in the high-order cases is addressed in the following. 

\subsection{HOCBF with respect to the Set \eqref{eqn-3}}
\label{subsec-discussion}

To ensure the FI property of the set $\mathbb{S}$ in \eqref{eqn-3}, we following Section \ref{subsec-HOCLF} to propose another novel HOCBF as follows. 

\begin{definition}
\label{def-6}
Consider the system \eqref{eqn-1} and the set $\mathbb{S}$ in \eqref{eqn-3}. A function $B\in\mathfrak{C}^{r}(\inte(\mathbb{S}), \mathbb{R})$ is called a \emph{high-order control barrier function (HOCBF-II)}, if for all $x\in\inte(\mathbb{S})$, item (i) of Definition \ref{def-4} holds and 
\begin{enumerate}[(i)]
  \item  there exist $\alpha_{1}, \alpha_{2}\in\mathcal{K}_{e}$ and $\Theta=(\theta_{1}, \ldots, \theta_{r})\in\mathbb{R}^{r \times r}$ such that
  \begin{align}
  \label{eqn-16}
  \alpha_{1}(h(x))\leq\frac{1}{B(x)}&\leq\alpha_{2}(h(x)), \\
  \label{eqn-17}
  \inf\nolimits_{u\in\mathbb{U}}\left\{L_{f}\mathbf{B}(x)+L_{g}\mathbf{B}(x)u\right\}&\preceq\Theta^{\top}\mathbf{B}(x), 
  \end{align}
  where $\mathbf{B}(x):=(B(x), L_{f}B(x), \ldots, L^{r-1}_{f}B(x))\in\mathbb{R}^{r}$;
  \item 
  there exist orthogonal matrices $\Omega:=(\omega_{1}, \ldots, \omega_{p})\in\mathbb{R}^{p\times p}, \Pi:=(\pi_{1}, \ldots, \pi_{r})\in\mathbb{R}^{p\times p}$ and a Jordan matrix $\bar{\Theta}\in\mathbb{R}^{r\times r}$ such that $\Theta=\Omega^{\top}\bar{\Theta}\Pi$ and $\sum^{r}_{i=1}\omega^{\top}_{1}\pi_{i}\mathbf{B}_{i}(x_{0})>0$.
\end{enumerate}
\end{definition}

Different from  Definition \ref{def-4} for the set $\bar{\mathbb{S}}$ in \eqref{eqn-6}, Definition \ref{def-6} is for the set $\mathbb{S}$ in \eqref{eqn-3}, which is larger than the one in \eqref{eqn-6}. 

\begin{theorem}
\label{thm-3}
Consider the system \eqref{eqn-1} admitting an HOCBF-II $B\in\mathfrak{C}^{r}(\inte(\mathbb{S}), \mathbb{R})$ with the set $\mathbb{S}$ in \eqref{eqn-3}. The system \eqref{eqn-1} is safe with respect to $\mathbb{S}$ under the following controller
\begin{align}
\label{eqn-18}
u(x):=\left\{\begin{aligned}
&0, &&\text{if\ } x=0, \\
&\kappa(\bar{\mathfrak{c}}(x), \bar{\mathfrak{d}}(x)), &&  \text{otherwise}, 
\end{aligned}\right.
\end{align}
where $\bar{\mathfrak{c}}(x):=L^{r}_{f}B(x)-\theta^{\top}_{r}\mathbf{B}(x), \bar{\mathfrak{d}}(x):=L_{g}L^{r-1}_{f}B(x)$, and $\kappa$ is defined in \eqref{eqn-11}.
\end{theorem}

\begin{IEEEproof}
Following the similar fashion as in the proof of Theorem \ref{thm-2}, $\dot{\mathbf{B}}(x)\preceq\Theta\mathbf{B}(x)$ holds for all $x\in\inte(\mathbb{S})$. From item (ii) of Definition \ref{def-6}, $B(x(t))\leq e^{\lambda_{\max}(\bar{\Theta})t}\sum^{r}_{i=1}\omega^{\top}_{1}\pi_{i}\mathbf{B}_{i}(x_{0})=:\beta(B(x_{0}), t)$. Since $e^{\lambda_{\max}(\bar{\Theta})t}>0$ for any $\lambda_{\max}(\bar{\Theta})\in\mathbb{R}$, $\beta(B(x_{0}), t)>0$ holds for all $t>0$. From \eqref{eqn-16}, $h(x(t))>\alpha^{-1}_{2}(1/\beta(B(x_{0}), t))>0$ holds for all $t>0$. As a result, the system \eqref{eqn-1} is safe with respect to the set $\mathbb{S}$. 
\end{IEEEproof}

From Theorem \ref{thm-3}, $\mathbb{S}$ in \eqref{eqn-3} can be a safe set for the system \eqref{eqn-1} and thus a larger safe set can be ensured in the high-order cases. Regarding the two proposed HOCBFs, some further discussions are presented. First, the differences of the HOCBF-I and HOCBF-II are discussed below. The HOCBF-I is based on the lasso-like functions in \eqref{eqn-7} and thus a small safe set in \eqref{eqn-6} is ensured. The HOCBF-II is based on the techniques in Section \ref{subsec-HOCLF} and the safe set in \eqref{eqn-3} can be guaranteed via an additional condition on the initial state (i.e., item (ii) of Definition \ref{def-6}). As a result, both the HOCBF-I and HOCBF-II can be available, and which is to be chosen depends on the considered system. Second, similar to Definition \ref{def-5}, the assumption on the matrix $\Theta\in\mathbb{R}^{r\times r}$ is made to ensure the positiveness of $B(x)$ along the time line. Just like \eqref{eqn-14}, a special case of the matrix $\Theta$ is 
\begin{align}
\label{eqn-19}
\Theta=\begin{bmatrix}
\theta^{\top}_{1}\\\vdots \\ \theta^{\top}_{r}
\end{bmatrix}=\begin{bNiceArray}{c|cw{c}{1cm}c}[margin]
0& \Block{3-3}{I} & &  \\
\vdots & & &  \\
0 & & & \\
\hline
\theta_{r1} &\theta_{r2} & \cdots& \theta_{rr}
\end{bNiceArray}. 
\end{align}
In this special case, $\Theta$ is required to be Hurwitz, which will be discussed later. Finally, besides the comparison principles in the proofs of Theorems \ref{thm-2}-\ref{thm-3}, an alternative technique is based on the properties of $\lambda_{p}\in\mathbb{R}^{p}$ or $\theta_{r}\in\mathbb{R}^{r}$; see \cite{Xu2018constrained}. For the two special cases as in \eqref{eqn-14} and \eqref{eqn-19}, the following result shows the equivalence relation between the two methods. 

\begin{proposition}
\label{prop-1}
Consider the matrix $\Lambda$ in \eqref{eqn-14}. The following two statements are equivalent. 
\begin{enumerate}
  \item The polynomial $y^{p}+\lambda_{p1}y^{p-1}+\ldots+\lambda_{pp}=0$ has negative roots only. 
  \item The matrix $\Lambda$ is Hurwitz. 
\end{enumerate}
\end{proposition}

\begin{IEEEproof}
From the first statement, let the roots of the polynomial be $-y_{1}, \ldots, -y_{p}<0$. We define the linear system 
\begin{align*}
\dot{\xi}&=\Lambda\xi, 
\end{align*}
and the notations 
\begin{align*}
z_{1}&=\xi_{1}, \quad z_{i}=\dot{\xi}_{i-1}+y_{i-1}\xi_{i-1}, \quad  i=2, \ldots, p. 
\end{align*}
The relation between $z$ and $\xi$ is given below 
\begin{align*}
z=U\xi, \quad  \dot{z}&=Kz, 
\end{align*}
where 
\begin{align*}
U&=\begin{bmatrix}
1 &0&0&\ldots&0&0\\ y_{1} &1&0&\ldots&0&0 \\y_{1}y_{2} &y_{1}+y_{2}&1&\ldots&0&0  \\ \vdots&\vdots&\vdots&\ddots&\vdots&\vdots \\  \prod^{p-1}_{i=1}y_{i} &\ldots&\ldots&\ldots&\sum^{p-1}_{i=1}y_{i}&1
\end{bmatrix},
\end{align*}
\begin{align*}
K&=\begin{bmatrix}
-y_{1} &1 &0&\ldots&0&0\\ 0& -y_{2} &1&\ldots&0&0\\ \vdots&\vdots&\vdots&\ddots&\vdots &\vdots\\  0 &0&0&\ldots&-y_{p-1} &1 \\ 0 &0&0&\ldots&0&-y_{p} 
\end{bmatrix}.
\end{align*}
Note that $U$ is  invertible and $K=U\Lambda U^{-1}$. Hence, the eigenvalues of $\Lambda$ and the roots of the polynomial are the same, which implies that $\Lambda$ is Hurwitz. In a similar way, $\Lambda=U^{-1}KU$ and the second statement can be guaranteed to be satisfied from the first statement. 
\end{IEEEproof}

The same result for the matrix $\Theta$ in \eqref{eqn-19} can be derived in the similar mechanism. From Proposition \ref{prop-1}, we can see the transformation of checking the roots of the polynomial into checking the Hurwitz property of a matrix. In the proof of Proposition \ref{prop-1}, the decomposition of the matrix $\Lambda$ is derived and thus $\bar{\Lambda}=K, \Omega=U$ and $\Pi=U^{-1}$. However, how to implement such a decomposition is flexible and can be set manually such that the constraint on $\mathbf{V}(x_{0})$ can be relaxed. 

\section{Optimal Control Design}
\label{sec-optcon}

Once the HOCBF and HOCLF are proposed, a direct way is to combine both of them to address the safety and stabilization control problems simultaneously. In this respect, we can implement the stabilizing controller in Theorem \ref{thm-2} to formulate the following optimization problem. 
\begin{align}
\label{eqn-20}
\begin{aligned}
\min&\quad  0.5|u-u_{\mathsf{s}}|^{2} \\
\text{s.t.}&\quad \mathfrak{c}(x)+\mathfrak{d}(x)u\leq0,
\end{aligned}
\end{align}
where $u_{\mathsf{s}}$ is the stabilizing controller \eqref{eqn-15}. In this way, the condition \eqref{eqn-9} is embedded in \eqref{eqn-20} such that the safety is always guaranteed. If the HOCBF-II is applied, then it is $\bar{\mathfrak{c}}(x)+\bar{\mathfrak{d}}(x)u\leq0$ that is embedded in \eqref{eqn-20}, and in this case the following mechanism can be implemented in a similar way. 

Next we only need to address how to modify the stabilizing controller in an optimal manner. Here we stress that the controller \eqref{eqn-15} may not be optimal, and $u_{\mathsf{s}}$ can be selected as certain optimal stabilizing controller \emph{a priori}. From \eqref{eqn-20}, the optimal controller is designed as
\begin{align}
\label{eqn-21}
u(x):=u_{\mathsf{s}}(x)+\bar{u}^{\ast}(x), 
\end{align}
where $\bar{u}^{\ast}(x)\in\mathbb{R}^{m}$ is the term to be designed. For this purpose, the following definition is introduced as an extension of \cite[Def. 2]{Cortez2022compatibility} to the high-order cases. 

\begin{definition}
\label{def-7}
Consider the system \eqref{eqn-1}. Let $\varkappa>0$ and $M(x)\in\mathbb{R}^{m\times m}$ be a locally Lipschitz continuous and positive-definite matrix. A set $\mathbb{B}(\varkappa)$ is said to be \emph{HOCBF-stabilizable}, if there exists $\mathbb{D}\subset\mathbb{B}(\varkappa)\cap\overbar{\mathbb{S}}$ such that 
\begin{align}
\label{eqn-22}
\mathfrak{b}(x)M^{-1}(x)\mathfrak{d}^{\top}(x)&\geq0, \quad \forall x\in(\mathbb{B}(\varkappa)\cap\overbar{\mathbb{S}})\setminus\mathbb{D}, \\
\label{eqn-23}
\tilde{\mathfrak{c}}(x)&\leq0, \quad \forall x\in\mathbb{D}, 
\end{align}
where $\tilde{\mathfrak{c}}(x):=\mathfrak{c}(x)+\mathfrak{d}(x)u_{\mathsf{s}}(x)$, $\mathfrak{b}(x)$ is defined below \eqref{eqn-15}, and $\mathfrak{c}(x), \mathfrak{d}(x)$ are defined below \eqref{eqn-10}.
\end{definition}

\begin{lemma}
\label{lem-1}
Consider the system \eqref{eqn-1}. For any locally Lipschitz continuous and positive-definite matrix $M(x)\in\mathbb{R}^{m\times m}$, there exists $\varkappa>0$ such that the set $\mathbb{B}(\varkappa)$ is HOCBF-stabilizable.
\end{lemma}

\begin{IEEEproof}
Since the origin is included in $\overbar{\mathbb{S}}$, we consider a region around the origin in $\overbar{\mathbb{S}}$. That is, there exists $\delta>0$ such that $\mathbb{B}(\delta)\subset\overbar{\mathbb{S}}$. Hence, $\phi_{r-1}(x)>0$ for all $x\in\mathbb{B}(\delta)$. From the Lipschitz continuity of the functions $f, g$ in \eqref{eqn-1} and the continuity of $u_{\mathsf{s}}(x)$, we have $|f(x)+g(x)u_{\mathsf{s}}(x)|\leq L|x|$, where $x\in\mathbb{B}(\delta)$ and $L>0$ is the Lipschitz constant. 

We can find a sufficiently small $\epsilon>0$ such that for all $x\in\mathbb{B}(\epsilon), \eta_{r}(1/\phi_{r-1}(x))>\epsilon L|\partial\phi_{r-1}(x)/\partial x|$. Let $\varkappa:=\min\{\epsilon, \delta\}$, and thus $\mathbb{B}(\varkappa)\subseteq\mathbb{B}(\delta)\subset\overbar{\mathbb{S}}$. For all $x\in\mathbb{B}(\varkappa)$, $|f(x)+g(x)u_{\mathsf{s}}(x)|<L\varkappa$ and $|\partial\phi_{r-1}(x)/\partial x||f(x)+g(x)u_{\mathsf{s}}(x)|\leq\epsilon L|\partial \phi_{r-1}(x)/\partial x|<\eta_{r}(1/\phi_{r-1}(x))$. From the norm inequality, $\tilde{\mathfrak{c}}(x)\leq|\partial\phi_{r-1}(x)/\partial x||f(x)+g(x)u_{\mathsf{s}}(x)|-\eta_{r}(1/\phi_{r-1}(x)<0$. Hence, there exists $\mathbb{D}\subset\mathbb{B}(\varkappa)\cap\overbar{\mathbb{S}}$ such that $\tilde{\mathfrak{c}}(x)\leq0$ for all $x\in\mathbb{D}$. That is, the set $\mathbb{B}(\varkappa)$ is HOCBF-stabilizable, and the proof is completed. 
\end{IEEEproof}

\begin{figure*}[!t]
\centering
\begin{tabular}{ccc}
\hspace{-5pt}\includegraphics[height=0.6\columnwidth]{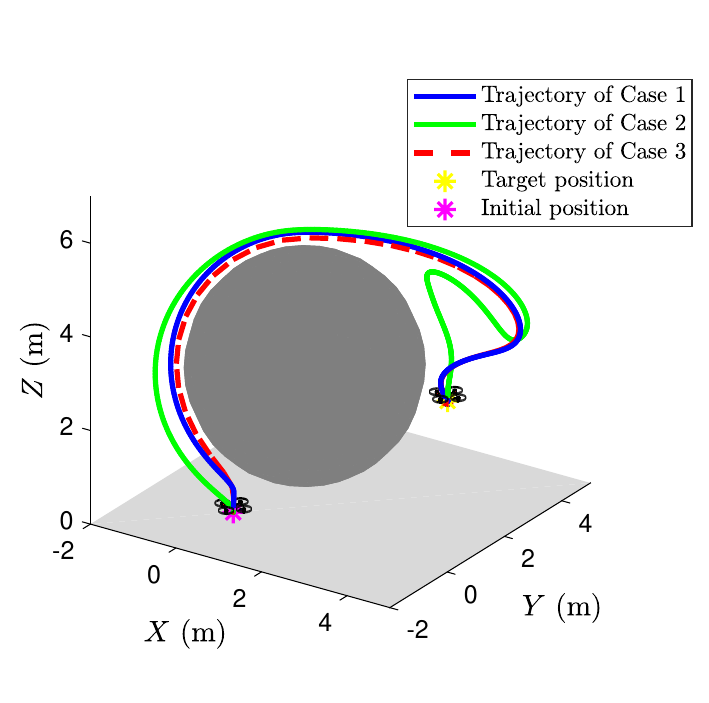}\hspace{-5pt}\vspace{-2pt} &
\includegraphics[height = 0.6\columnwidth]{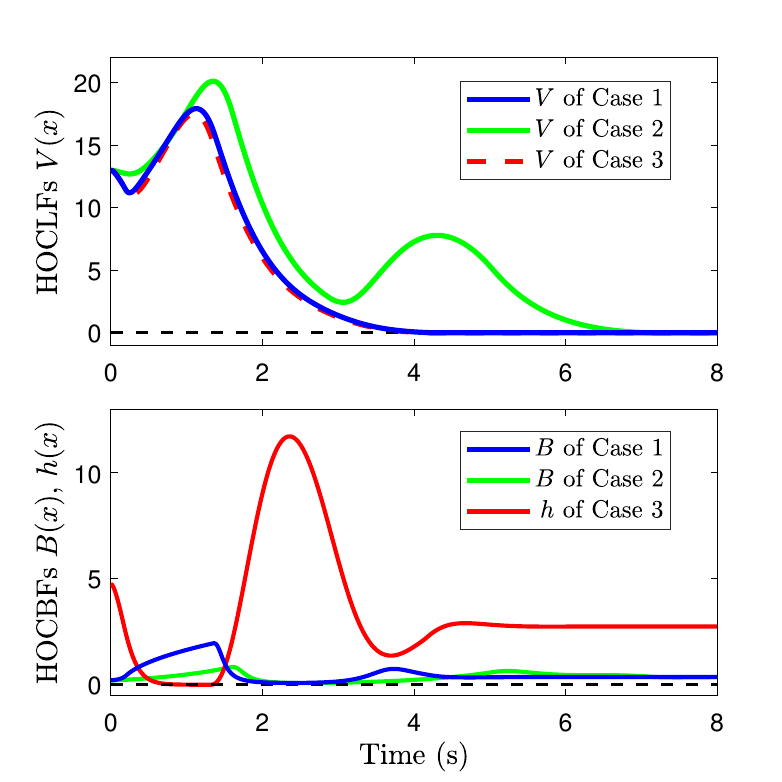}\hspace{-5pt}\vspace{-2pt} &
\includegraphics[height = 0.6\columnwidth]{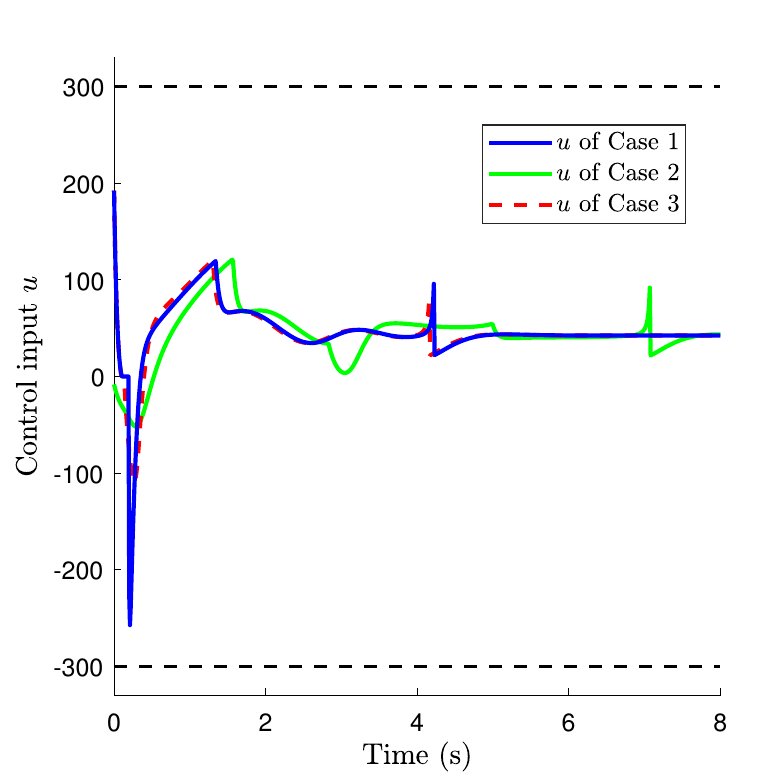}\vspace{-2pt} \\
\hspace{-5pt}\tiny{(a)}\hspace{-5pt}\vspace{-2pt} & \tiny{(b)}\hspace{-5pt}\vspace{-2pt} & \tiny{(c)}\vspace{-2pt} 
\end{tabular}
\caption{Illustration of the safe navigation of the quadrotor in the three cases. (a) The position trajectories of the quadrotor. (b) The evolution of the HOCLFs and HOCBFs. (c) The control inputs for the quadrotor.}
\label{fig1}
\end{figure*}

From Lemma \ref{lem-1}, the HOCBF-stabilizable set does exist. With the HOSCBF-stabilizable set, we can establish the optimal solution to \eqref{eqn-20} and derive the controller \eqref{eqn-21}. 

\begin{theorem}
\label{thm-4}
Consider the system \eqref{eqn-1}, the set $\overbar{\mathbb{S}}\in\mathbb{R}^{n}$ in \eqref{eqn-6}, and the problem \eqref{eqn-20}. The controller \eqref{eqn-21} with 
\begin{align}
\label{eqn-24}
\bar{u}^{\ast}(x):=\left\{\begin{aligned}
&0, &&\text{if\  } \tilde{\mathfrak{a}}(x)\leq0 \\
&-\frac{\tilde{\mathfrak{c}}(x)}{|\mathfrak{d}(x)|^{2}}\mathfrak{d}^{\top}(x),  &&  \text{otherwise}
\end{aligned}\right.
\end{align}
 is an optimal solution to \eqref{eqn-20}. Under the controller \eqref{eqn-21} with \eqref{eqn-24}, the set $\overbar{\mathbb{S}}$ is FI, and the closed-loop system is SGAS with respect to the HOCBF-stabilizable set. 
\end{theorem}

\begin{IEEEproof}
First, we prove that the controller \eqref{eqn-21} with \eqref{eqn-24} is the optimal solution to \eqref{eqn-20}. From \cite[Ch. 5]{Boyd2004convex}, the problem \eqref{eqn-20} is convex with respect to the control input, and the solution to \eqref{eqn-20} exists if and only if the Karush-Kuhn-Tucker (KKT) conditions are satisfied. The Lagrangian for \eqref{eqn-20} is defined as 
\begin{align*}
\mathfrak{L}(x, u, \varrho)&:=0.5|u-u_{\mathsf{s}}|^{2}+\varrho(\mathfrak{a}(x)+\mathfrak{b}(x)u), 
\end{align*}
where $\varrho\geq0$ is a Lagrange multiplier. The KKT conditions are given below.
\begin{align}
\label{eqn-25}
\frac{\partial\mathfrak{L}(x, u, \varrho)}{\partial u}:=u-u_{\mathsf{s}}+\varrho\mathfrak{d}^{\top}(x)&=0, \\
\label{eqn-26}
\varrho(\mathfrak{c}(x)+\mathfrak{d}(x)u)&=0.
\end{align}
From \eqref{eqn-25}-\eqref{eqn-26}, if $\mathfrak{c}(x)+\mathfrak{d}(x)u=0$, then no constraint is imposed to the choice of $\varrho\geq0$. In this case, $\varrho$ is chosen to be $0$ such that the cost function in \eqref{eqn-20} is minimized, thereby resulting in the optimal solution $u^{\ast}(x)=u_{\mathsf{s}}(x)$. If $\mathfrak{c}(x)+\mathfrak{d}(x)u\neq0$, then $\varrho\equiv0$ from \eqref{eqn-26} and thus $u^{\ast}(x)=u_{\mathsf{s}}(x)$ from \eqref{eqn-25}. In this case, we need to consider two scenarios. The first scenarios is that $\tilde{\mathfrak{c}}(x)\leq0$, which implies that the safety constraint holds. Hence, $u^{\ast}(x)=u_{\mathsf{s}}(x)$ minimizes the cost function. The second scenario is that $\tilde{\mathfrak{c}}(x)>0$, and the safety constraint does not hold. Furthermore, $u_{\mathsf{s}}$ is not the solution to \eqref{eqn-20} and a novel controller is needed. In this scenario, from \eqref{eqn-25}-\eqref{eqn-26}, we have $\varrho=\tilde{\mathfrak{c}}(x)|\mathfrak{d}(x)|^{-2}$ and thus $u^{\ast}(x)=u_{\mathsf{s}}(x)-\tilde{\mathfrak{c}}(x)|\mathfrak{d}(x)|^{-2}\mathfrak{d}^{\top}(x)$. Hence, the controller \eqref{eqn-21} with \eqref{eqn-24} is the optimal solution. 

Next, from \eqref{eqn-24}, if $\tilde{\mathfrak{c}}(x)\leq0$, then $\bar{u}^{\ast}(x)=0$ and 
\begin{align}
\label{eqn-27}
&L_{f}\phi_{r-1}(x)+L_{g}\phi_{r-1}(x)u_{\mathsf{s}}(x)\leq\eta_{r}(1/\phi_{r-1}(x)). 
\end{align}
If $\tilde{\mathfrak{c}}(x)>0$, then $u(x)=u_{\mathsf{s}}+\bar{u}^{\ast}(x)$, and further  
\begin{align}
\label{eqn-28}
&L_{f}\phi_{r-1}(x)+L_{g}\phi_{r-1}(x)(u_{\mathsf{s}}(x)+\bar{u}^{\ast}(x)) \nonumber \\
&=L_{f}\phi_{r-1}(x)+L_{g}\phi_{r-1}(x)u_{\mathsf{s}}(x)-\tilde{\mathfrak{c}}(x)  \nonumber \\
&=\eta_{r}(1/\phi_{r-1}(x)). 
\end{align}
From \eqref{eqn-27}-\eqref{eqn-28} and Theorem \ref{thm-1}, we conclude that under the controller \eqref{eqn-21} with \eqref{eqn-24}, the set $\overbar{\mathbb{S}}$ is FI. 

Finally, let $M(x):=\tilde{\mathfrak{c}}(x)|\mathfrak{d}(x)|^{-2}I$, and from Lemma \ref{lem-1} there exists an HOCBF-stabilizable set $\mathbb{B}(\varkappa)$ with $\varkappa>0$. Since $u(x)=u_{\mathsf{s}}(x)+\bar{u}^{\ast}(x)$, we have 
\begin{align*}
&L^{p}_{f}V(x)+L_{g}L^{p-1}_{f}V(x)u(x)+\lambda^{\top}\mathbf{V}(x) \nonumber \\
&\leq-\sqrt{\mathfrak{a}^{2}(x)+|\mathfrak{b}(x)|^{4}}+L_{g}L^{p-1}_{f}V(x)\bar{u}^{\ast}(x)  \nonumber \\
&\leq L_{g}L^{p-1}_{f}V(x)\bar{u}^{\ast}(x). 
\end{align*}
For all $x\in\mathbb{B}(\varkappa)$, $\bar{u}^{\ast}(x)=0$ if $\tilde{\mathfrak{c}}(x)\leq0$. Otherwise, from \eqref{eqn-22} we have $L_{g}L^{p-1}_{f}V(x)\bar{u}^{\ast}(x)=-\mathfrak{b}(x)M(x)\mathfrak{d}^{\top}(x)\leq0$. Hence, for all $x\in\mathbb{B}(\varkappa)$, $L^{p}_{f}V(x)+L_{g}L^{p-1}_{f}V(x)u(x)+\lambda^{\top}\mathbf{V}(x)\leq0$, which implies from the proof of Theorem \ref{thm-2} that the closed-loop system is SGAS with respect to $\mathbb{B}(\varkappa)$. 
\end{IEEEproof}

Theorem \ref{thm-4} presents the optimal controller for both stabilization and safety objectives. The proposed controller extends the classic case in \cite{Cortez2022compatibility}  to the high-order cases, and does no involve the choice of the gain margin in \cite{Jankovic2018robust}. In addition, if the HOCBF-II is applied, then a similar controller can be derived to guarantee the FI property of the set \eqref{eqn-3}. 

\section{Numerical Results}
\label{sec-example}

In this section, a numerical example from the safe navigation of quadrotors is presented to illustrate the derived results. All computation is executed via MATLAB R2023a on a laptop with AMD Ryzen 9 5900HX (3.30GHz) and 16GB RAM. Consider the quadrotor modelled in the inertial frame \cite{lee2010geometric}:
\begin{equation}
\label{eqn-29}
\begin{aligned}
\dot{p}&=v, && M\dot{v}=Mg\bm{e}_3-\mathbf{R}\bm{e}_3F, \\
\dot{\mathbf{R}}&=\mathbf{R}\hat{\boldsymbol{\omega}}, &&\mathbf{J}\dot{\boldsymbol{\omega}}=-\boldsymbol{\omega}\times\mathbf{J}\boldsymbol{\omega}+\boldsymbol{\tau}, 
\end{aligned}
\end{equation}
where $p\in\mathbb{R}^3$ is the position and $v\in\mathbb{R}^3$ is the velocity. $M>0$ is the mass of quadrotor, $g$ is the acceleration of gravity, $\bm{e}_3=[0, 0, 1]^\top$, and $\mathbf{R}\in\mathsf{SO}(3)$ is the rotation matrix from the body-fixed frame to the inertial frame, where $\mathsf{SO}(3):=\{\mathbf{R}\in\mathbb{R}^{3\times3}: \mathbf{R}^\top\mathbf{R}=I, \mathsf{det}(\mathbf{R})=1\}$ and $\mathsf{det}(\mathbf{R})$ is the determinant of $\mathbf{R}$. $F\in\mathbb{R}$ is the thrust, which is taken as the control input. In the body-fixed frame, $\boldsymbol{\omega}\in\mathbb{R}^3$ is the angular velocity, $\hat{\boldsymbol{\omega}}\in\mathsf{SO}(3)$ is the skew symmetric matrix derived from $\boldsymbol{\omega}$, $\mathbf{J}\in\mathbb{R}^{3\times3}$ is the inertia matrix, ``$\times$'' is the cross product, and $\boldsymbol{\tau}\in\mathbb{R}^3$ is the moment; see \cite{lee2010geometric} for more details. 

The quadrotor aims to execute a safe navigation mission in an environment; see Fig. \ref{fig1}(a). The unsafe set is defined as $\mathbb{O}:=\{p\in\mathbb{R}^3: (p-p_{o})^2\leq r_{o}^2\}$, which is a spherical obstacle. $p_{o}\in\mathbb{R}^3$ and $r_{o}>0$ are respectively the center and radius of the obstacle. To maintain an appropriate safety margin, let the safe set be $\mathbb{S}:=\{p\in\mathbb{R}^3: h(p)\geq0\}$ with 
\begin{equation} 
\label{eqn-30}
  h(p):=(p-p_{o})^2-r^2, 
\end{equation}
where $r:=r_{o}+\delta$ with a safety margin $\delta>0$. Since $h(p)$ is only dependent on the position $p$, the attitude controller in \cite{lee2010geometric} is applied directly such that only the position controller needs to be designed via the optimization problem \eqref{eqn-20}. 

In order to complete the safe navigation mission, the first goal is to reach a desired position $p_{d}\in\mathbb{R}^3$. To this end, we define the following HOCLF candidate:
\begin{equation} 
\label{eqn-31}
V(p):=(p-p_{d})^2.
\end{equation}
Hence, the position controller is to guarantee the convergence of $V(p)$. From Definition \ref{def-5}, $V$ is an HOCLF if the conditions \eqref{eqn-12}-\eqref{eqn-13} are satisfied. The second goal is to guarantee the quadrotor to be navigated safely. From \eqref{eqn-30}, we introduce the following HOCBF candidate:
\begin{equation}
\label{eqn-32}
  B(p):=1/h(p),
\end{equation}
which is an HOCBF-I if the conditions \eqref{eqn-8}-\eqref{eqn-9} hold and is an HOCBF-II if the conditions \eqref{eqn-16}-\eqref{eqn-17} hold. 

Here we consider the following three cases. Case 1 is to combine the proposed HOCLF and HOCBF-I. Case 2 is to combine the  HOCLF and HOCBF-II, while Case 3 is to combine the HOCLF and the HOCBF in \cite{Tan2021high}. Let $M=4.34$, $g=9.8$, $\mathbf{J}=\mathsf{diag}(0.082, 0.0845, 0.1377)$, $p_{o}:=[1, 1, 3]^\top$, $p_{d}:=[3, 3, 2]^\top$, $r_{o}=2.35$ and $\delta=0.15$. For the HOCLF \eqref{eqn-31}, let $\lambda=[-10, -11]^{\top}$. For the HOCBF-I \eqref{eqn-32}, let $\eta_0(s):=2s$ and $\eta_1(s):=10s$ for all $s>0$. For the HOCBF-II, let $[\Theta_{21},\Theta_{22}]=[10,-10]$. For the HOCBF in \cite{Tan2021high}, let $\gamma_1(s):=10s$ for all $s>0$. 

Based on \eqref{eqn-20}-\eqref{eqn-21} and \eqref{eqn-24}, all the simulation results are presented in Fig. \ref{fig1}. From Fig. \ref{fig1}(a), the obstacle avoidance is achieved in the three cases. From Fig. \ref{fig1}(b), $B(p)>0$ in Cases 1 and 2, and $h(p)>0$ in Case 3, all of which correspond to the obstacle avoidance in Fig. \ref{fig1}(a). The HOCLF $V(p)$ in the three cases converges to $0$ asymptotically, which implies the satisfaction of the navigation mission. However, for Cases 1 and 3 it only takes 4s to achieve the navigation mission, while it takes around 6.5s for Case 2. The control inputs of the three cases are presented in Fig. \ref{fig1}(c), where we can see the bounds of the control input in Case 2 are minimal. 

To compare the control efforts in the three cases, we introduce the evaluation function $U_i:=\int_{0}^{8}|u_i(t)|dt$, $i\in\{1, 2, 3\}$ and derive $U_{1}=400.92, U_{2}=383.18, U_{3}=399.28$, which in turn shows the advantages of the HOCBF-II. If the control input is constrained, then the explicit controller \eqref{eqn-24} cannot be derived, and the optimization problem \eqref{eqn-20} needs to be solved in real time. Let $F\in[-300, 300]$, and the computation times for the three cases are $t_{1}=6.11$ms, $t_{2}=6.87$ms, and $t_{3}=6.26$ms. In this way, we can see the metrics of the HOCBF-I in terms of the computational complexity. 

\section{Conclusions}
\label{sec-conclusion}

In this paper, we considered the safety and stabilization control problems of nonlinear systems via high-order control barrier and Lyapunov functions. A novel reciprocal high-order control barrier function was proposed for the safety control, while novel high-order control Lyapunov function was proposed for the stabilization control. The compatibility of the safety and stabilization was investigated such that these two objectives could be achieved in a unified manner. Future work will focus on high-order control functions for more general cases like distributed, stochastic and time-delay cases.

\bibliographystyle{IEEEtran}
\bibliography{bibfile26}
\end{document}